\documentclass[copyright,creativecommons]{eptcs}

\usepackage{amssymb}
\usepackage{amsmath}
\usepackage{amsthm}
\usepackage{graphicx}
\usepackage{color}
\usepackage{epsfig}
\usepackage{breakurl}

\providecommand{\event}{PACO 2011}
\providecommand{\firstpage}{1}

\newcommand{\shorten}[1]{}

\newcommand{\AT}{\textsf{AT}}

\newcommand{\GAC}{\widehat{\mbox{\tt AC}}}

\newcommand{\GST}{\widehat{\mbox{\tt ST}}}
\newcommand{\GTR}{\widehat{\mbox{\tt TR}}}

\newcommand{\TR}{\mbox{\tt TR}}

\newcommand{\at}{\mbox{\tt at}}

\newcommand{\man}{\texttt{crule}{!}}  
\newcommand{\ok}{\mbox{\tt ok}}

\newcommand{\tr}{\mbox{\tt tr}}
\newcommand{\trap}{\texttt{crule}{?}}   

\newcommand{\trprime}{\tr'}

\newcommand{\sscheck}{\mbox{\scriptsize\sl check}}

\newcommand{\sstriv}{\mbox{\scriptsize\sl triv}}

\newcommand{\Sprime}{S \mkern1mu {}'}

\newcommand{\tarrow}{\arrow{}{t}}

\newcommand{\trivarrow}{\xra{\sstriv}}

\newcommand{\sscontinue}{\mbox{\scriptsize\sl continue}}
\newcommand{\ssdone}{\mbox{\scriptsize\sl done}}

\newcommand{\ssnotyet}{\mbox{\scriptsize\sl notYet}}

\newcommand{\sspermit}{\mbox{\scriptsize\sl permit}}

\newcommand{\ssrefuse}{\mbox{\scriptsize\sl refuse}}
\newcommand{\ssrequest}{\mbox{\scriptsize\sl request}}

\newcommand{\checkiarrow}{\xra{\sscheck_i}}
\newcommand{\continuearrow}{\xra{\sscontinue}}
\newcommand{\donearrow}{\xra{\ssdone}}

\newcommand{\notyetarrow}{\xra{\ssnotyet}}

\newcommand{\permitarrow}{\xra{\sspermit}}

\newcommand{\refusearrow}{\xra{\ssrefuse}}
\newcommand{\requestarrow}{\xra{\ssrequest}}

\newcommand{\Client}{\textsf{Client}}
\newcommand{\PAClient}{\widehat{\textsf{Client}}}
\newcommand{\QClient}{\textsf{QClient}}
\newcommand{\PAQClient}{\widehat{\textsf{QClient}}}
\newcommand{\RClient}{\textsf{Client}\mkern1mu{}'}
\newcommand{\RClientQuotient}{\textsf{QClient}'}
\newcommand{\PARClient}{\widehat{\textsf{Client}}\mkern2mu{}'}
\newcommand{\SClient}{\textsf{Client}\mkern1mu{}''}
\newcommand{\SClientQuotient}{\textsf{QClient}''}
\newcommand{\PASClient}{\widehat{\textsf{Client}}{}''}
\newcommand{\PASClientQuotient}{\widehat{\textsf{QClient}}{}''}
\newcommand{\CS}{\textsf{CS}}

\newcommand{\DG}{\textsf{DG}}

\newcommand{\LabelsD}{\textsf{Labels}_D}

\newcommand{\NDetServer}{\textsf{Server}}
\newcommand{\OK}{\textsf{OK}}
\newcommand{\PANDetServer}{\widehat{\textsf{Server}}}

\newcommand{\Server}{\textsf{Server}}

\newcommand{\StatesD}{\textsf{States}_D}

\newcommand{\AtDoor}{\textsf{AtDoor}}
\newcommand{\Busy}{\textsf{Busy}}
\newcommand{\Idle}{\textsf{Idle}}

\newcommand{\Interrupt}{\textsf{Interrupt}}
\newcommand{\Out}{\textsf{Out}}
\newcommand{\Waiting}{\textsf{Waiting}}
\newcommand{\With}{\textsf{With}}
\newcommand{\Without}{\textsf{Without}}

\newcommand{\NDHelping}{\textsf{NDHelping}}
\newcommand{\NDChecking}{\textsf{NDChecking}}

\newcommand{\checki}{\textsf{check}}
\newcommand{\continue}{\textsf{continue}}

\newcommand{\done}{\textsf{done}}
\newcommand{\enter}{\textsf{enter}}
\newcommand{\return}{\textsf{return}}
\newcommand{\explain}{\textsf{explain}}

\newcommand{\leave}{\textsf{leave}}
\newcommand{\notYet}{\textsf{notYet}}

\newcommand{\permit}{\textsf{permit}}

\newcommand{\refuse}{\textsf{refuse}}
\newcommand{\request}{\textsf{request}}
\newcommand{\thank}{\textsf{thank}}
\newcommand{\triv}{\textsf{triv}}

\newcommand{\merge}{\mathbin{\parallel}}
\newcommand{\mycom}{\mathbin{|}}
\newcommand{\xra}[1]{\xrightarrow{#1}}
\newcommand{\bisim}
  {\;\mathrel{\underline{\hspace*{-0.15ex}\leftrightarrow\hspace*{-0.15ex}}}\;}
\newcommand{\bbisim}{\bisim_{\!\!b\ }}

\mathcode`:="003A       
\mathcode`;="203B       
\mathcode`?="003F       
\mathcode`|="026A       
\mathcode`<="4268       
\mathcode`>="5269       

\mathchardef\ls="213C   
\mathchardef\gr="213E   
\mathchardef\uparrow="0222      
\mathchardef\downarrow="0223    

\newcommand{\blankline}{\vspace*{\baselineskip}}
\newcommand{\halflineup}{\vspace*{-0.5\baselineskip}}
\newcommand{\nop}[1]{}

\newcommand{\ttcheck}{\mbox{\tt check}}

\newcommand{\AC}{\mbox{\tt AC}}

\newcommand{\ST}{\mbox{\tt ST}}

\newcommand{\TS}{\mbox{\tt TS}}

\newcommand{\ac}{\mbox{\tt ac}}

\newcommand{\st}{\mbox{\tt st}}

\newcommand{\acprime}{\ac'}

\newcommand{\stprime}{\st'}

\newcommand{\arrow}[2]{\,{{\stackrel{#2}{\rightarrow}}_{#1}}\,}

\newcommand{\aarrow}{\arrow{}{a}}

\newcommand{\lc}{\lbrace\:}

\newcommand{\rc}{\:\rbrace}

\newtheorem{theorem}{Theorem}
\newtheorem{lem}[theorem]{Lemma}

\newtheorem{definition}[theorem]{Definition}

\title{Towards reduction of Paradigm coordination models}
\author{%
  Suzana Andova\footnote{Corresponding author,
    email~\url{s.andova@tue.nl}.}
  \institute{Department of Mathematics and Computer Science \\
    TU/e, Eindhoven, the Netherlands}
  \and
  Luuk Groenewegen
  \institute{FaST Group, Leiden Institute of Advanced Computer
    Science \\ Leiden University, The Netherlands}
  \and
  Erik de Vink
  \institute{Department of Mathematics and Computer Science \\
    TU/e, Eindhoven, the Netherlands}
}
\def\titlerunning{Towards reduction of Paradigm coordination models}
\def\authorrunning{S. Andova, L.P.J.~Groenewegen and E.P. de Vink}

\begin{document}

\maketitle

\begin{abstract}
  \textbf{Abstract}
  The coordination modelling language Paradigm addresses collaboration
  between components in terms of dynamic constraints. Within a
  Paradigm model, component dynamics are consistently specified at a
  detailed and a global level of abstraction. To enable automated
  verification of Paradigm models, a translation of Paradigm into
  process algebra has been defined in previous work.  In this paper we
  investigate, guided by a client-server example, reduction of
  Paradigm models based on a notion of global inertness.
  Representation of Paradigm models as process algebraic
  specifications helps to establish a property-preserving equivalence
  relation between the original and the reduced Paradigm
  model. Experiments indicate that in this way larger Paradigm models can
  be analyzed.

  \smallskip


\end{abstract}


\section{Introduction}
\label{sec:1}

Within the current software architecture practice, architectures are
mostly used for describing static aspects of software
systems. Techniques that allow system architects to describe
coordination among components within an architecture and to reason
about the dynamics of the system in its entirety, are not commonly
used. The coordination description language Paradigm helps the
designer to merge different dynamic aspects of a system.  At the same
time the language allows for the description of both detailed and
global behaviour of an individual component i.e.\ its own specific
behaviour and separately its interaction with other components, and
the language is particularly helpful in enforcing consistency in the
behaviour of large sets of interrelated components.

The coordination modeling language Paradigm~\cite{GV02,coord06}
specifies roles and interactions within collaborations between
components. Interactions are in terms of temporary constraints on the
dynamics of components. To underpin Paradigm models with formal
verification and automated analysis, the Paradigm language has been
linked with the \texttt{mCRL2} toolset~\cite{JFGDagstuhl} via its
translation to the process algebra ACP~\cite{BBR10,scp10} and with the
probabilistic modelchecker Prism~\cite{KNP09,isola2010} via a direct
encoding scheme. Process algebras (PA for short), such as CCS, CSP,
LOTOS and ACP, provide a powerful framework for formal modeling and
reasoning about concurrent systems, which turns out to be very
suitable for our needs in the setting of coordination. The key
concepts of compositionality and synchronization in process algebra
are mostly exploited in our translation. As detailed and global
aspects of component behaviour are specified by separate PA
specifications, the vertical constraints are encoded through
synchronizations expressing consistency of detailed and global
component behaviour. Horizontal constraints at the protocol level are
naturally captured by parallel composition, synchronization and
encapsulation.

While the translation to ACP and \texttt{mCRL2} allows for formal
verification of Paradigm models~\cite{scp10,ads6,isola2010}, the
omnipresent problem of state space explosion when analyzing large
models occurs here as well. In the present paper, we address the
question of reducing Paradigm models of coordination. The reduction
method applies to a component's behaviour, reducing the representation
of the vertical constraints of that component by abstracting away any
information on the component behaviour irrelevant for these
constraints. To this end, the benefit of the translation of Paradigm
language into ACP is twofold. On the one hand, we borrow the
abstraction concept from PA and apply it directly in Paradigm on
detailed behaviour. On the other hand, the translation provides us
with a formal proof methodology to reason and guarantee that the
reduced Paradigm model has the same properties as the original model.
As a matter of fact, it has gradually become evident that separating
detailed from global behaviour as supported by the Paradigm language,
allows us to reason about reduction by abstraction in a rather natural
way. We shall clarify this point after the Paradigm overview, at the
end of Section~\ref{sec:2}.

Our work on dynamic consistency in a horizontal and vertical dimension
has been influenced by the work of
K\"{u}ster~\cite{EHMG02,kuester}. Related work includes the Wright
language~\cite{All97:phd} based on CSP provides FDR support to check
both types of consistency properties. Other bridges from software
architecture to automated verification include the pipeline from UML
via Rebeca and Promela to the SPIN model-checker and from UML via
Object-Z and CSP to the FDR
model-checker~\cite{SMSB05,MORW08}. Process algebra driven prototyping
as coordination from CCS is proposed in~\cite{RB05}. The skeletons
generated from CCS-specifications overlap with Paradigm
collaborations.  In the TITAN framework~\cite{PNMC07}, CCS is playing
a unifying role in a heterogeneous environment for aspect-oriented
software engineering. Recently the coordination language Reo has been
equipped with a process algebraic
interpretation~\cite{Arb04:mscs,sefm2010}. The encoding of Reo into
\texttt{mCRL2} and subsequent analysis has been integrated in the ECT
toolset for Reo~\cite{Kra11:phd}.

We present our idea by means of an example. The system we consider
consists of~$n$ clients who try to get service from one server
exclusively, a critical section problem, where the server is supposed
to choose the next client in a non-deterministic manner.  While the
translation of the Paradigm model into PA for the example is done
manually, the toolset \texttt{mCRL2} is exploited to generate the
complete state spaces, on which further analysis can be done. Initial
results show a substantial reduction in the size of the state
space. In Section~\ref{sec:2} Paradigm is summarized on the basis of
the above example.  Section~\ref{sec:3} briefly introduces our process
algebra translation for the example model.  In Section~\ref{sec:4} we
present our reduction techniques.  Section~\ref{sec:5} concludes the
paper.


\section{Paradigm and  a critical section model}
\label{sec:2}

\noindent
This section briefly describes the central notions of Paradigm: STD,
phase, (connecting) trap, role and consistency rule.
\begin{itemize}
\item An \emph{STD} $Z$ (\emph{state-transition diagram}) is a triple
  $Z=<\ST, \AC, \TR>$ with $\ST$~the set of states, $\AC$~the set of
  actions and $\TR \subseteq \ST \times \AC \times \ST$ the set of
  transitions of $Z$,  notation $x \aarrow x'$.
\item A \emph{phase} $S$ of an STD $Z=<\ST,\AC,\TR>$ is
  an STD $S=<\st,\ac, \tr>$ such that $\st \subseteq \ST$, $\ac
  \subseteq \AC$ and $\tr \subseteq \lc (x,a,x') \in \TR \mid x,x' \in
  \st, a \in \ac \rc$.
\item A \emph{trap}~$t$ of phase $S=<\st,\ac,\tr>$ of STD $Z$ is a non-empty
  set of states $t \subseteq \st$ such that $x \in t$ and $x \,
  \aarrow \, x' \in \tr$ imply $x' \in t$.
  A trap~$t$ of phase $S$ of STD $Z$ \emph{connects}
  phase~$S$ to a phase $\Sprime = <\stprime, \acprime,
  \trprime>$ of $Z$ if $t \subseteq \stprime$. Such trap-based
  connectivity between
  two phases of $Z$ is called a \emph{phase transfer} and is denoted
  as $S \tarrow \Sprime$.
\item A \emph{partition} $\pi=\lc (S_i,T_i) \mid i \in I \rc$ of an
  STD $Z=<\ST,\AC,\TR>$, $I$ a non-empty index set, is a set of pairs
  $(S_i,T_i)$ consisting of a phase $S_i=<\st_i,\ac_i, \tr_i>$ of $Z$
  and of a set $T_i$ of traps of $S_i$.
\item A \emph{role} at the level of a partition $\pi = \lc (S_i,T_i)
  \mid i \in I \rc$ of an STD $Z= <\ST,\AC,\TR>$ is an STD $Z(\pi) =
  <\GST,\GAC,\GTR>$ with $\GST \subseteq \lc S_i \mid i \in I \rc$,
  $\GAC \subseteq \bigcup_{i \in I} T_i$ and $\GTR \subseteq \lc S_i
  \tarrow S_j \mid i,j \in I, t \in \GAC \rc$ a set of phase
  transfers. $Z$ is called the \emph{detailed} STD underlying
  \emph{global} STD~$Z(\pi)$, being role $Z(\pi)$.
\item A \emph{consistency rule} or \emph{protocol step} for an
  ensemble of STDs $Z, Z_1 , \ldots , Z_k$ and roles $Z_1(\pi_{1}) ,
  \ldots , Z_k(\pi_{k})$ is a nonempty set of phase
  transfers preceded by
  one extra transition.
\item Let $Z \colon x \aarrow x' \ast Z_1(\pi_{1}) \colon S \mkern1mu
  {}_{1}' \tarrow S \mkern1mu {}_{1}'', \dots, Z_k(\pi_{k}) \colon S
  \mkern1mu {}_{k}' \tarrow S \mkern1mu {} _k''$ be a consistency
  rule for a given ensemble; $Z_i , \ldots , Z_k$ are
  \emph{participants} of it, $Z$ is \emph{conductor}.
\item A Paradigm \emph{model} is an ensemble of STDs, roles thereof
  and consistency rules.
\end{itemize}

\noindent
The above notions constitute Paradigm models.  The semantics thereof
are roughly as follows: a consistency rule has synchronization of its
phase transfers and its conductor transition, only if all connecting
traps mentioned have been entered. Detailed transitions are allowed in
the current state of an STD, only if the current phase (state) of each
role of the STD contains the transition. In this way, phases are
constraints on underlying STD dynamics imposed by protocols (sets of
protocol steps).
In a mirrored way, traps impose constraints on the behaviour at the
protocol level, as traps are involved in the firing of consistency
rules.

\begin{figure}[hbt!]
\vspace{-0.5cm}
  \begin{center}
  \epsfig{file=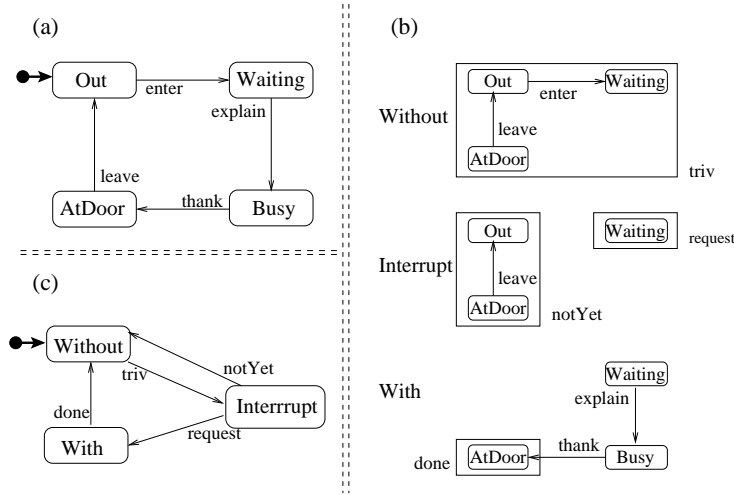,scale=0.50}
  \end{center}
\vspace{-0.5cm}
  \caption{(a)~detailed STD of $\Client$, (b) partition of three
    phases, (c)~global STD $\Client(\CS)$.}
  \label{visualsyntax1}
\end{figure}

\noindent
An STD is a step-wise description of the dynamics belonging to a
component.  It is visualized as a directed graph: its nodes are
states, its action-labeled edges are transitions. Initial states are
graphically indicated by a black
dot-and-arrow. Figure~\ref{visualsyntax1}a gives the so-called
detailed STD of a $\Client$ in and around a shop: starting in state
$\Out$ the client cycles through states $\Waiting$, $\Busy$, $\AtDoor$
and $\Out$ again, subsequently.  The entire system we consider,
contains $n$~such clients, dynamically the same, plus one different
component, the server.  For the complete system the overall
requirement is that only one client at a time, out of all $n$~clients,
is allowed to be in its state $\Busy$. So, being in state $\Busy$ is a
Critical Section problem (abbreviated $\CS$).  To solve it, ongoing
$\Client_i$ dynamics is constrained by the phase prescribed
currently. Figure~\ref{visualsyntax1}b visualizes phases $\Without$,
$\Interrupt$ and $\With$. Phase $\Without$ excludes being in state
$\Busy$ by prohibiting to take the actions $\explain$ and
$\thank$. Contrarily, phase $\With$ allows both, going to and leaving
state $\Busy$. Finally, the intermediate phase $\Interrupt$ is an
interrupted form of $\Without$, as action $\enter$ cannot be taken,
but being in state $\Waiting$ is allowed, though.

In view of a transfer from the current phase into a next phase to occur,
enough progress within the current phase must have been made: a
connecting trap has to be entered first. Figure~\ref{visualsyntax1}b
pictures relevant connecting traps for the above three phases, drawn
as rectangles around the states the trap consists of. In particular, we
need trap $\triv$ to be connecting from $\Without$ to $\Interrupt$,
trap $\notYet$ to be connecting from $\Interrupt$ back to $\Without$,
trap $\request$ to be connecting from $\Interrupt$ forward to $\With$
and finally, trap $\done$ to be connecting from $\With$ back to
$\Without$. In this manner, Figure~\ref{visualsyntax1}b gives all
ingredients needed for the dynamics of a $\Client_i$ STD at the
level of partition $\CS$: see role $\Client_i(\CS)$
in Figure~\ref{visualsyntax1}c and repeated in Figure~\ref{visualsyntax2}a.

\begin{figure}[bht!]
\vspace{-0.5cm}
  \begin{center}
  \centerline{\epsfig{file=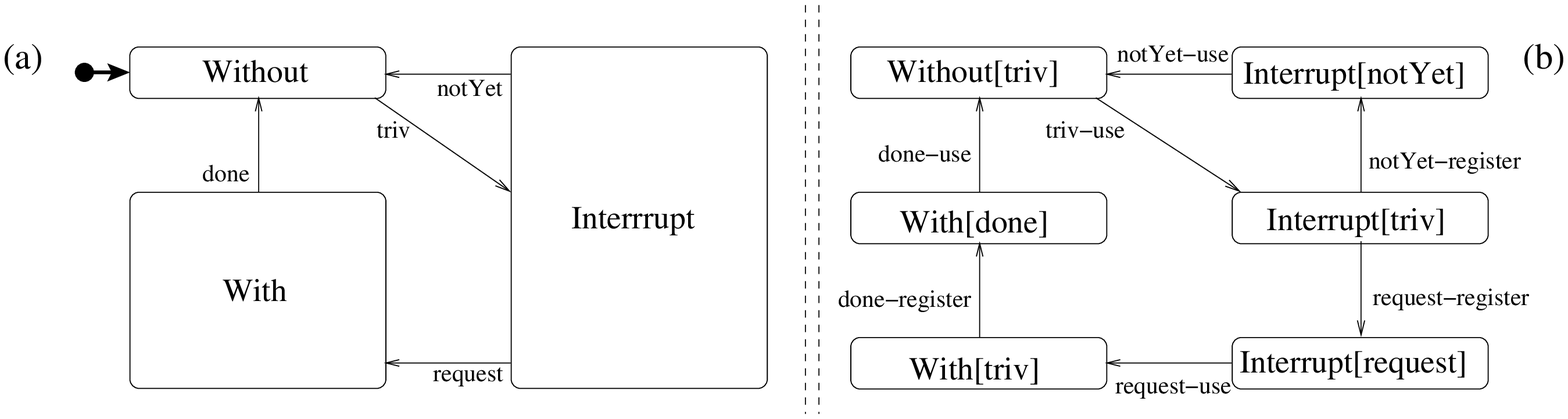,scale=0.50}}
  \end{center}
\vspace{-0.8cm}
  \caption{(a)~global process $\Client(\CS)$ and (b)~its refinement
    in view of translation.}
  \label{visualsyntax2}
\end{figure}

\noindent
Figure~\ref{visualsyntax2}b presents a slightly refined diagram of the
proper role STD in part~(a). State names here, additionally keep track
of the trap most recently entered within a phase, as if it could be
taken as a smaller phase committed to within the larger one imposed.
Action names still refer to a trap that is entered, but they
additionally discriminate between, first, \emph{registering} the trap
has been entered and, second, thereafter \emph{using} this for a phase
transfer.  This more refined view represents the starting point for
the ACP encoding of the global process, as discussed in the next
section.

So far, we have discussed `sequential composition' of constraints:
imposed phases alternated with traps committed to. Semantically, any
current phase constrains the enabled transitions to those belonging to
the phase. So, at any moment a current detailed state belongs to the
current phase too. From this it follows, that the dynamics of the
detailed STD and of the global STD are consistent, the current global
phase reflects the current local state. Paradigm's consistency rules
are to the essence of `parallel composition': they express coupling of
role steps of arbitrarily many \emph{participants} and a detailed step
of one \emph{conductor}. Any consistency rule specifies the
simultaneous execution of the steps mentioned in the rule, a
transition of the conductor and phase transfers for the participants.

To continue the example of $n$~clients getting service, one at a time,
we present a non-deterministic coordination solution for the
$n$~clients via a server. The non-deterministic server checks the
clients in arbitrary order. If a client, when checked, wants help, it
gets help by being permitted to enter the critical section. If not,
permission to enter is refused to it. Only after a client's leaving
the critical section, the server stops helping it by returning to the
idle position, from which it arbitrarily selects a next client for
checking. In the example, the server provides a unique conductor step
for each consistency rule.  The STD $\NDetServer$ of the server is
drawn in Figure~\ref{conductorND}. As conductor, detailed steps of
$\NDetServer$ need to be coupled to phase transfers of each
$\Client_i$, $1 \leq i \leq n$.

\begin{figure}[thb]
  \begin{center}
    \input{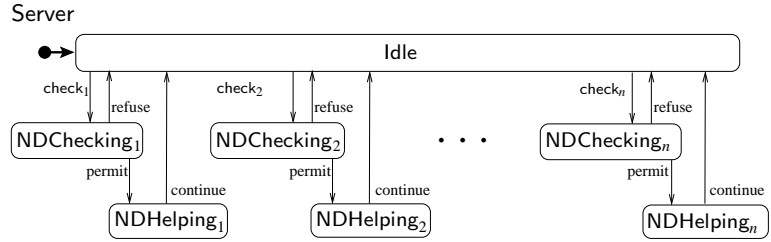}
  \end{center}
\vspace{-0.5cm}
  \caption{STD non-deterministic server $\NDetServer$.}
  \label{conductorND}
\end{figure}

  \fontsize{9pt}{0pt}
\begin{align}
    &
    \NDetServer \mbox{ } \colon \mbox{ }  \Idle  \checkiarrow
    \NDChecking_i  \mbox{ }
    \ast
    \mbox{ } \Client_i(\CS) \mbox{ } \colon  \Without  \trivarrow
    \Interrupt
    \label{cs01}
    \\[0.1cm]
    &
    \NDetServer \mbox{ } \colon \mbox{ }  \NDChecking_i  \refusearrow
    \Idle  \mbox{ }
    \ast
    \mbox{ } \Client_i(\CS) \mbox{ } \colon  \Interrupt
    \notyetarrow  \Without
    \\[0.1cm]
    &
    \NDetServer \mbox{ } \colon \mbox{ }  \NDChecking_i  \permitarrow
    \NDHelping_i  \mbox{ }
    \ast
    \mbox{ }  \Client_i(\CS) \mbox{ } \colon  \Interrupt
    \requestarrow  \With
    \\[0.1cm]
    &
    \NDetServer \mbox{ } \colon \mbox{ }  \NDHelping_i  \continuearrow
    \Idle  \mbox{ }
    \ast
    \mbox{ } \Client_i(\CS) \mbox{ } \colon  \With  \donearrow
    \Without
\end{align}

\noindent
Note that for this protocol, each conductor step of the server
corresponds to a phase change of exactly one client. E.g., the server
moves from the state $\Idle$ to $\NDChecking_i$ iff the global client
process $\Client_i(\CS)$ changes from the phase $\Without$ to the
phase $\Interrupt$. The server then makes a $\ttcheck_i$
transition. In general, there is a precondition, however. Within the
phase $\Without$ sufficient progress should have been made, such that
the particular trap has been reached. In this case, it is the trivial
trap $\triv$ rendering the requirement superfluous, as the trivial
trap, containing all states of the phase $\Without$, is trivially
reached. For the actual checking, the next two consistency rules,
dependent on the trap $\notYet$ and $\request$, respectively, decide
the target of the conductor transition and the next participant phase,
viz.\ state $\Idle$ and phase $\Without$ or state $\NDHelping_i$ and
phase $\With$, respectively. The last consistency rule couples the
conductor's returning from state $\NDHelping_i$ to $\Idle$ with trap
$\done$ of phase $\With$ having been entered.

The consistency rules specify \emph{horizontal} dynamic consistency,
i.e.~across components, here between server and clients. Such
specification is about coordination, i.e.~what Paradigm actually
models, step-wise computation of next behavioural constraints.  The
constraining property imposed by a phase implies, an underlying
$\Client_i$ transition is allowed only if it belongs to the phase that
corresponds to the current state of the role of $\Client_i$ in the
$\CS$ collaboration, i.e.\ the current state of the global STD
$\Client_i(\CS)$. The constraining property $\Client_i$ commits to by
entering a trap, allows for a phase transfer, i.e.\ a transition of
$\Client_i(\CS)$, once the (connecting) trap is entered.  These two
constraining properties syntactically guarantee \emph{vertical}
dynamic consistency, i.e.~within a component between its underlying
STD and its role.

As mentioned in Section~\ref{sec:1}, is has become evident to us that
separating detailed from global behaviour as supported by the Paradigm
language, allows one to reason about reduction by abstraction in a
rather natural way.  The intuitive explanation for this is as follows:
Global behaviour, actually defining phases a system needs to go
through during a particular coordination solution, is built on top of
the detailed behaviour: each global phase represents a sub-behaviour
of the underlying detailed behaviour. Nevertheless, not every action
at the detailed level affects the current global phase. Only some
actions may enable a next phase transfer and hence may affect the
protocol execution. Thus, it is natural to try to detect the detailed
actions that do not matter for, i.e.~that cannot be observed at, the
protocol level.  By hiding them, a reduced detailed behaviour is
obtained, just containing all relevant information and actions needed
for proper execution of the component role within the protocol. As we
shall show for our running example, this information can be extracted
from the hierarchical structure per component in the Paradigm model,
see Subsection~\ref{subsec:4-1}.  Note that all interaction between
components (horizontal) and all hierarchical structure within
components (vertical), as specified in the Paradigm model in an
explicit manner, are flattened in the PA translation and hence their
character being either horizontal or vertical, gets lost.  Thus, after
the PA translation only a single communication pattern remains, from
which it is no longer straightforward to extract information needed
for proper reduction of detailed behaviour.

Yet another aspect of the Paradigm model that can be justified and
confirmed by the approach taken here is discussed shortly in the
paper, see Subsection~\ref{subsec:4-2}. From the definition of
Paradigm, although provided with a formal operational semantics, it is
not straightforward to see to what extent a component's detailed
behaviour is not affected by some constraints or coordination rule.
In particular, consistency rules for some complex model may have an
unforeseen effect on detailed component behaviour, in particular a
deadlock at the detailed level.  The translation from Paradigm to ACP
combined with the abstraction techniques discussed in the next section
supports formal verification of separate protocols and of overall
coordination.


\section{Paradigm model as a process algebraic specification}
\label{sec:3}

In this section we show by means of the example introduced in
Section~\ref{sec:2}, how a Paradigm model can be translated into
ACP\@. The general translation has been defined in~\cite{scp10} to
which we refer for more detail. Roughly, each STD will be represented
by a recursive specification.  Vertical consistency in Paradigm has to
be expressed explicitly.  In particular, to represent the interaction
of a detailed STD and the global STD, we use actions $\ok!(.)$ and
$\ok?(.)$ that take the labels of \emph{detailed steps} as their
argument.  The complementary actions synchronize if the step of the
detailed STD is allowed by the current phase of the global STD as
constraint. Thus, synchronization of actions $\ok!(\cdot)$
and~$\ok?(\cdot)$ between global STD and detailed STD reflect the
current permission for the detailed step to be taken.

In addition, we use the complementary actions $\at!(.)$ and $\at?(.)$
that take \emph{detailed states} as their arguments. The complementary
actions synchronize if the step to be taken by the global STD is
allowed by the current trap of the detailed STD as constraint.  Upon
synchronization of $\at!(\cdot)$ and~$\at?(\cdot)$ the global process
will update its trap information, if applicable.  For the
communication within the protocol, here between the server and its
clients, actions $\man(.)$ on the side of a conductor are meant to
complement $\trap(.)$ actions on the side of the
employees. Synchronization leads to execution of the corresponding
consistency rule: a detailed transition of the conductor, phase
changes for the employees involved.

For the concrete example the above amounts to the following. We adorn
the $n$~processes $\Client_i$ with the actions~$\at!$, conveying state
information, and actions~$\ok?$, regarding transition eligibility.
\begin{displaymath}
  \fontsize{9pt}{0pt}
  \def\arraystretch{1.1}
  \begin{array}{@{}rcl}
    \PAClient_i & = & \Out_i
    \\
    \Out_i & = & \at ! (\Out_i) \cdot \Out_i
    + \ok ? (\enter_i) \cdot \Waiting_i
    \\
    \Waiting_i & = & \at ! (\Waiting_i) \cdot \Waiting_i
    + \ok ? (\explain_i) \cdot \Busy_i
    \\
    \Busy_i & = & \ok ? (\thank_i) \cdot \AtDoor_i
    \\
    \AtDoor_i & = & \at ! (\AtDoor_i) \cdot \AtDoor_i
    + \ok ? (\leave_i) \cdot \Out_i
  \end{array}
  \def\arraystretch{1.0}
\end{displaymath}
The LTS of~$\PAClient_i$ of~$\Client_i$ is given in
Figure~\ref{basicACPspecifications}a (with the subscript~$i$
suppressed).
The definition of process $\PAClient_i$ assures, the process really
starts in close correspondence to starting state $\Out$ from
Figure~\ref{visualsyntax1}a. The definition of process $\Out_i$
expresses: (1) upon being asked, it can exchange state information
while keeping the process as-is; (2) it can ask for permission to take
the analogue of transition $\enter$ from Figure~\ref{visualsyntax1}a,
in view of continuing with process $\Waiting_i$ thereafter. Note, in
the definition of process $\Busy_i$ the possibility for exchange of
state information is not specified, as asking for it does never occur.
Note, in Figure~\ref{visualsyntax1}b, state $\Busy$ does not belong to
trap $\done$.

\begin{figure}[bht!]
  \centering
  \input{Figures/clientAlla.pstex_t}
  \caption{Processes (a)~$\PAClient$ and (b)~$\PAClient(\CS)$.}
  \label{basicACPspecifications}
\end{figure}

In a similar manner, the~$n$ processes $\Client_i(\CS)$ are augmented
with the actions $\at?$ and~$\ok!$. Now, at the global level, the
relevant information is the pair of the current phase and the current
trap. For example, the recursion variable $\Without_i[\triv]$
represents that $\Client_i$ is constrained to phase $\Without$ and
hasn't reached a specific trap, whereas $\Interrupt_i[\notYet]$
reflects that $\Client_i$ committed to phase $\Interrupt$ resides in
trap $\notYet$. As these global processes play a participant role in the
protocol, the $\trap$ actions for engaging in a consistency rule have
been put in place as well.

\begin{displaymath}
  \fontsize{9pt}{0pt}
  \def\arraystretch{1.1}
  \begin{array}{@{}rcl}
    \PAClient_i(\CS) & = & \Without_i[\triv]
    \\
    \Without_i[\triv] & = &
    \ok ! (\leave_i) \cdot \Without_i[\triv] +
    \ok ! (\enter_i) \cdot \Without_i[\triv] +
    {} \\  &  &
    \trap  (\triv_i) \cdot \Interrupt_i[\triv]
    \\
    \Interrupt_i[\triv] & = &
    \at ? (\AtDoor_i) \cdot \Interrupt_i[\notYet] +
    \at ? (\Out_i) \cdot \Interrupt_i[\notYet] +
    {} \\ &  &
    \at ? (\Waiting_i) \cdot \Interrupt_i[\request] +
    \ok ! (\leave_i) \cdot \Interrupt_i[\triv]
    \\
    \Interrupt_i[\notYet] & = &
    \ok ! (\leave_i) \cdot \Interrupt_i[\notYet] +
    \trap (\notYet_i) \cdot \Without_i[\triv]
    \\
    \Interrupt_i[\request] & = &
    \trap (\request_i) \cdot \With_i[\triv]
    \\
    \With_i[\triv] & = &
    \at ? (\AtDoor_i) \cdot \With_i[\done] +
    \ok ! (\explain_i) \cdot \With_i[\triv] +
    {} \\  & &
    \ok ! (\thank_i) \cdot \With_i[\triv]
    \\
    \With_i[\done] & = &
    \trap (\done_i) \cdot \Without_i[\triv]
  \end{array}
  \def\arraystretch{1.0}
\end{displaymath}

\noindent
The corresponding LTS
of the specification $\PAClient_i(\CS)$ of $\Client_i(\CS)$ is given in
Figure~\ref{basicACPspecifications}b.

As above, process $\PAClient_i(\CS)$ is defined in close
correspondence to $\Without_i[\triv]$ being starting state in
Figure~\ref{visualsyntax2}b. The $\ok!(.)$-actions provide the
permission answers to requests from $\PAClient_i$ to take a detailed
step. The $\at?(.)$-actions ask for state information relevant for
deciding a next, smaller trap has been entered. The $\trap(.)$-actions
correspond to a phase change, so they synchronize with a particular
conductor step.

The final component of the Paradigm model that needs to be translated
into ACP is the non-determinis\-tic server $\NDetServer$. In fact, the
STD of the server as given in Figure~\ref{conductorND} exactly
corresponds to its recursive specification; we only rename each
transition label~$\ell$ from Figure~\ref{conductorND} into
$\man(\ell)$ to stay consistent with the general translation as
defined in~\cite{scp10}, for instance $\permit_i$ is renamed into
$\man(\permit_i)$ in the PA specification. There is neither any
$\ok(.)$ action nor any $\at(.)$ action added here. This component
plays the conductor role in the protocol and as such it is represented
only by its detailed behaviour (detailed STD). Therefore, no vertical
constraints are imposed on its detailed behaviour.

\begin{displaymath}
  \fontsize{9pt}{0pt}
  \def\arraystretch{1.1}
  \begin{array}{@{}rcl}
    \PANDetServer & = & \Idle
    \\
    \Idle & = &
    \man (\checki_1) \cdot \NDChecking_1
    \, + \cdots + \,
    \man (\checki_n) \cdot \NDChecking_n
    \\
    \NDChecking_i & = & \man (\permit_i) \cdot \NDHelping_i +
    \man (\refuse_i) \cdot \Idle
    \\
    \NDHelping_i &  = & \man (\continue_i) \cdot \Idle
  \end{array}
  \def\arraystretch{1.0}
\end{displaymath}

\noindent
For the communication function~`$|$' we put $\at!(s) \mycom \at?(s) =
\tau$ for `states' $s = \Out_i , \Waiting_i , \AtDoor_i $, and
$\ok?(a) \mycom \ok!(a) = \ok(a)$, for actions $a = \enter_i
,\explain_i , \thank_i , \leave_i$. Note, ACP allows to keep the
result of the synchronization of $\ok?(a)$ and~$\ok!(a)$
observable, here as the action $\ok(a)$, for suitable~$a$. We exploit
this feature below to express system properties, since the
synchronization actions $\ok(a)$ describe detailed steps taken by
clients. E.g., observing $\ok(\enter_i)$ indicates a service request
made by $\Client_i$. On the contrary, synchronization of $\at!()$ and
$\at?()$ is only used to update the information of the current detailed
state. The resulting actions are internal to the component and not
needed in any further analysis. Therefore, we safely use~$\tau$ for
the synchronization of $\at?()$ and $\at!()$.

Finally, we need to encode the coordination captured by the
consistency rules. For example, consistency rule~(\ref{cs01}) couples
a detailed $\checki_i$ step of the $\Server$, being the conductor of the
$\CS$~protocol, to the global $\triv$ step of $\Client_i$, being a
participant in the $\CS$~protocol. The net result is a state transfer,
i.e.\ a transition $\Idle \checkiarrow \NDChecking_i$ for the server,
and a phase transfer, i.e.\ a transition $\Without \trivarrow
\Interrupt$ in the global STD for the $i$-th client. Similar
correspondences apply to the other consistency rules. Therefore, we
put
\begin{displaymath}
  \fontsize{9pt}{0pt}
  \def\arraystretch{1.1}
  \begin{array}{r@{\,}c@{\,}l@{\;}c@{\;}l}
    \man (\checki_i) & | & \trap (\triv_i)
    & = & \checki_i
    \\
    \man (\permit_i) & | & \trap (\request_i)
    & = & \permit_i
  \end{array}
  \qquad
  \begin{array}{r@{\,}c@{\,}l@{\;}c@{\;}l}
    \man (\refuse_i) & | & \trap (\notYet_i)
    & = & \refuse_i
    \\
    \man (\continue_i) & | & \trap (\done_i)
    & = & \continue_i
  \end{array}
  \def\arraystretch{1.0}
\end{displaymath}
As usual, unmatched synchronization actions will be blocked to enforce
communication.  We collect those in the set $A = \lc \man, \trap,
\at?, \at!, \ok?, \ok! \rc$. Finally, the process for the
collaboration of the server and the $n$~clients is given by
\begin{equation}
  \def\arraystretch{1.3}
  \begin{array}[b]{@{}l}
    \partial_A( \, \PAClient_1 \merge \PAClient_1(\CS)
    \merge \ldots
    \merge
    \PAClient_n \merge \PAClient_n (\CS) \merge \PANDetServer \, )
  \label{NDsystem}
  \end{array}
  \def\arraystretch{1.3}
\end{equation}

\noindent
The next section is concerned with the intertwining of detailed and
the global behavior, and possible ways to reduce the component
specification by abstracting away from specific detailed
activities. The process algebraic specification of our running
client-server example will be used below to establish relations
between Paradigm models before and after reduction. Therefore, it
comes in handy to represent the overall behaviour of the $\Client$
component as the parallel composition of its detailed and global
behaviour. To this end, we denote the set of states of the detailed
process $\PAClient$ by $\StatesD = \lc \Out, \Waiting, \Busy, \AtDoor
\rc$, the set of labels of its transitions by of detailed $\LabelsD =
\lc \enter, \explain, \thank, \leave \rc$ and we put

\begin{displaymath}
    \AT = \lc \at!(s), \, \at?(s) \mid s\in \StatesD \rc
    \qquad
    \OK  = \lc \ok!(a), \, \ok?(a) \mid a \in \LabelsD \rc
\end{displaymath}
and define $H = \AT \cup \OK$.  Then the process combining detailed
behaviour of $\PAClient$ and global behaviour of $\PAClient(\CS)$ can
be expressed as $\PAClient(\DG)$, with $\DG$ referring to `detailed'
and `global', given by
\begin{displaymath}
  \PAClient(\DG) = \partial_{H}(\PAClient \merge \PAClient(\CS))
\end{displaymath}
Figure~\ref{parallel_client_and_clientCS} shows the behavior of
$\PAClient(\DG)$ graphically.  The process describes the way the
detailed and global behaviors occur and constrain each other.

\begin{figure}[bht!]
  \centering
  \scalebox{0.50}{\input{Figures/clientDG.pstex_t}}
  \caption{Process $\PAClient(\DG)$}
  \label{parallel_client_and_clientCS}
\end{figure}

On the one hand, steps taken at the detailed level influence the
current phase at the global level, and therefore allows and forbids
certain phase transitions at the global level. The global process and
its transitions, are `navigated' by the activities executed at the
detailed level. For instance, the effect of the detailed transition
$\ok(\enter)$ is described with the appearance of two $\triv$
transitions. One of them captures the scenario in which the client has
not yet required any service, which means that $\enter$ has not been
taken yet at the detailed level, although the server (conductor) may
offer service. It can be observed that this transition is followed by
the phase transition $\notYet$ which brings the process back to the
initial state. We can also observe that as soon as the detailed
transition $\enter$ is taken, the enabled $\triv$ transition differs
from the previous one.

On the other hand, from $\PAClient(\DG)$ we can observe how each
phase, i.e.\ a global state, constrains the steps that can be taken
locally. Moreover, it is specified exactly how a trap that is reached
blocks any detailed transitions, just as expected. For instance, we
see that the action $\ok(\leave)$ on top of
Figure~\ref{parallel_client_and_clientCS} cannot be executed before
the phase is changed, i.e.\ a step from $\With[\done]$ to
$\Without[\triv]$ via the global transition $\trap(\done)$. Note that
such details, which are explicit and easily observable from the ACP
specification of the composition $\PAClient(\DG)$, cannot be directly
detected in the Paradigm model.

Once systems are modeled algebraically, their behaviours can be
compared.  Comparison is typically done by means of equivalence
relations, chosen appropriately to preserve certain properties. Since
we aim at the mCRL2 toolset for tool support, we choose for branching
bisimulation~\cite{GW96} as the equivalence relation we apply. Indeed,
branching bisimulation is the strongest in the spectrum of behavioural
equivalence relations, but yet weak enough to identify sufficiently
many systems. Below we adapt the definition from~\cite{GW96}
(originally defined on labelled transition systems) to STDs with
uniquely indicated initial states. In fact, labelled transition
systems (LTS), as a (visual) representation of process algebraic
specifications, can be seen also as STDs. Therefore, in the sequel we
do not make explicit distinction between LTSs and STDs.

\blankline

\begin{definition}
  \label{branching_bisimulation}
  For two STDs $Z = <\ST, \AC, \TS>$, $Z' = <\ST', \AC', \TS'>$ a
  symmetric relation $R \subseteq \ST \times \ST'$ is called a
  branching bisimulation relation if for all $s \in \ST$ and $t \in
  \ST'$ such that $R(s,t)$, the following condition is met: if $s
  \xrightarrow{a} s'$ in~$Z$, for some $a\in\AC\cup\{\tau\}$, then
  either $a = \tau$ and $R(s',t)$, or for some $n \geq 0$, there exist
  $t_1, \ldots, t_n$ and~$t'$ in~$\ST'$ such that $t
  \xrightarrow{\tau} t_1 \xrightarrow{\tau} \ldots \xrightarrow{\tau}
  t_{n} \xrightarrow{a} t'$ in~$Z'$, $R(s,t_1), \ldots, R(s,t_n)$
  and~$R(s',t')$.
\end{definition}

\noindent
For two STDs $Z$ and~$Z'$, two states $s \in Z$ and~$t \in Z'$ are
called branching bisimilar, notation $s\bbisim t$, if there exists a
branching bisimulation relation~$R$ for $Z$ and~$Z'$ such
that~$R(s,t)$. The STDs $Z$ and $Z'$ are branching bisimilar, notation
$Z\bbisim Z'$ if their initial states are branching bisimilar.


\section{Reduction of the client processes}
\label{sec:4}

\noindent
In Section~\ref{sec:3} we explained how ACP specifications are
obtained from the detailed and global client STDs, and how ACP's
communication function captures synchronization of detailed and global
steps, guaranteeing consistent dynamics at both levels.  Based on the
complete client component we are able to make several observations
regarding the Paradigm approach to separate the detailed from the
global behaviour.

\subsection{First-reduce then-compose}
\label{subsec:4-1}

The global STD of a component is an abstract representation of its
detailed STD. It represents the part of the behaviour of the component
that is essential for the interaction within a given collaboration. In
general, for the global behaviour not all local transitions are
relevant, most are not influencing the overall coordination at
all. Although not always easy to isolate, in actual full-fledged
systems only a restricted part of the whole system provides a specific
functionality. In such a situation, from a modeling perspective it is
clarifying to abstract away the irrelevant part and to concentrate on
a reduced detailed behaviour containing the relevant interaction. As a
consequence, dealing with models that are purposely made concise 
becomes simpler, more feasible and less error-prone.

In the previous sections, we have made a Paradigm model out of the
components: detailed client STDs, their global STDs and the server
STD\@. Moreover, we have presented their translations into process
algebraic specifications. The overall behaviour of the client-server
system is obtained by putting the components involved in parallel and
make them interact. In this section we show that we can achieve the
same total behaviour of the client-server system by \emph{first
  reducing} the client components and \emph{then composing} the
reduced versions afterwards with other components of the
system. Reduction is directly applied on the original Paradigm client
model, by abstracting away irrelevant states and local transitions.

It is intuitively clear that the global behaviour alone is not
branching bisimilar to the overall client behaviour
$\PAClient(\DG)$. This is because some local steps change the further
global behaviour. As a consequence, such local transitions can be
detected at the global level. Extending terminology going back
to~\cite{GW96}, we call these transitions \emph{globally
  non-inert}. Similarly, a local transition is referred to as
\emph{globally inert} if it cannot be observed, explicitly or
implicitly, at the global level. More specifically, it can be detected
whether local action $\enter$ has been taken or not by observing
whether the global transition $\notYet$ or global transition
$\request$ follows after global step $\triv$. Putting it differently,
the transition labeled $\enter$ makes the difference for phase
$\Interrupt$ of residing in trap $\notYet$ or in trap $\request$, as
can be seen in Figure~\ref{visualsyntax1}. Thus, the local transition
$\enter$ is not globally inert. In a similar manner, the local action
$\thank$ is not globally inert as it enables --and so it can be
detected-- the execution of the global action $\done$. In terms of the
partition, in phase $\With$ the action $\thank$ enters the
trap~$\done$. On the other hand, again referring to the phases of
$\Client(\CS)$ in Figure~\ref{visualsyntax1}b, we see that the action
$\leave$ is in each phase either within a trap (phases $\Without$ and
$\Interrupt$) or not possible at all (phase $\With$ is missing the
target state $\Out$). Likewise, the action $\explain$ is not possible
(phases $\Without$ and $\Interrupt$ are missing state $\Busy$) or
doesn't change the trap information (in phase $\With$ the transition
doesn't enter the trap $\done$).

\begin{definition}
  Let a Paradigm model be given. A detailed transition $x \aarrow x'$
  of a participant of a protocol is called \emph{globally inert} with
  respect to its partition~$\pi = \lc (S_i,T_i) \mid i \in I \rc$ if
  for all traps~$t$ in~$T_i$ it holds that $x \in t \iff x' \in t$
  whenever both $x,x' \in S_i$, $i \in I$. An action $a$ is called
  globally inert for a participant of a protocol with respect to a
  partition, if all $a$-labeled transitions are.
\end{definition}

Using the notion of detailed transitions being globally inert or
non-inert, we can reduce the detailed STD of the client. After
renaming all globally inert transitions into $\tau$, we can identify
branching bisimilar states. The resulting quotient STD for the client
carries the behaviour that is necessary and sufficient for the global
STD to interact with the other components, including the conductor of
the collaboration. The composition of the process algebraic
specifications of the quotient STD and the global $\PAClient(\CS)$
behaves exactly (up to branching bisimulation) as the behaviour of the
composition of the original detailed and global STDs together as
represented by $\PAClient(\DG)$.  By congruence, composition of either
of these systems with the other clients and the server leads, modulo
branching bisimulation equivalence, to the same behaviour. This is
summarized by the next result, where $\tau_I$, for a set of
labels~$I$, represents the hiding of the actions in~$I$ from~$P$ by
renaming them into~$\tau$, and $\partial_J(P)$, for a set of
labels~$J$, is the encapsulation of the actions of~$J$ from~$P$ by
blocking and transition for~$P$ with label in~$J$.

\bigskip

\begin{lem}
  \label{lemma_reduce_compose}
  Let $G \subseteq \LabelsD$
  be a subset of globally inert actions.
  Then it holds for the induced quotient $\QClient$ of
  $\Client$ that
  \begin{enumerate}
  \item [(i)] $\QClient \bbisim \tau_G(\Client)$, and
  \item [(ii)] $\partial_{H} ( \, \PAQClient \merge \PAClient(\CS) \,
    ) \bbisim \tau_{\OK(G)}( \, \PAClient(\DG) \, )$,
    where $\OK(G) = \lc \ok(a) \mid a \in G \rc$.
  \end{enumerate}
\end{lem}

  \begin{figure}[bht!]
    \begin{center}
      \centerline{\input{Figures/clienthidden.pstex_t}}
      \caption{(a)~process $\tau_G(\Client)$ and related states,
        (b)~quotient STD $\QClient$ and (c)~$\PAQClient$.}
      \halflineup
      \label{client_partitioned}
    \end{center}
  \end{figure}

\begin{proof}
  We consider the case of the maximal set of local actions that are
  globally inert, i.e.\ for $G = \lc \explain, \leave \rc$. Split the
  set of states $\StatesD$ of the detailed STD into $P = \lc \Out, \,
  \AtDoor \rc$ and $Q= \lc \Waiting, \, \Busy \rc$.  Let $\QClient$ be
  the induced quotient STD, the STD obtained from $\Client$ by
  identifying the states $\Out$ and $\AtDoor$ as well as the states
  $\Waiting$ and $\Busy$. The processes $\QClient$ and
  $\tau_G(\Client)$ are shown in
  Figure~\ref{client_partitioned}ab.
    A branching bisimulation between $\QClient$ and
  $\tau_G(\Client)$ can be immediately established, which proves the
  first part of the lemma.

  In order to prove the second part of the lemma, we first translate
  $\QClient$ into the process algebraic specification $\PAQClient$
  whose STD is shown in Figure~\ref{client_partitioned}c.  In order to
  compute the composition of $\PAQClient$ and $\PAClient(\CS)$ the
  communication function has to be adapted to $\PAQClient$. For the
  $\PAQClient$ process $\Out$ and $\AtDoor$ are identified into the
  $P$. Similar for $\Waiting$, $\Busy$, now represented by~$Q$. Thus,
  a detailed $\PAQClient$ communication intention conveying
  `\emph{at~$P$}' or `\emph{at~$Q$}' updates the global process about
  the current local state. Hence, we extend the communication function
  with $\at!(P) \mycom \at?(\Out) = \tau$, $\at!(P) \mycom
  \at?(\AtDoor) = \tau$, $\at!(Q) \mycom \at?(\Waiting) = \tau$ and
  $\at!(Q) \mycom \at?(\Busy) = \tau$.  Now we consider the process
  $\partial_{H} ( \, \PAQClient \merge \PAClient(\CS) \, )$ with $H =
  \AT \cup \OK$ as defined in Section~\ref{sec:3}, with $\AT$ extended
  accordingly. The composition is shown in
  Figure~\ref{reduce_then_compose}a, the process $\tau_{\OK(G)}( \,
  \PAClient(\DG) \, )$ is depicted in
  Figure~\ref{reduce_then_compose}b. It is straightforward to
  establish a branching bisimulation between these two processes.
\end{proof}

\begin{figure}[bht!]
  \centerline{\input{Figures/clientDGfromreduced.pstex_t}\qquad\qquad\mbox{}}
    \caption{Branching bisimilar processes: (a)~$\partial_{H} ( \,
      \PAQClient \merge \PAClient(\CS))$ (b)~process
      $\tau_{\OK(G)}(\PAClient(\DG))$.}
    \label{reduce_then_compose}
\end{figure}

\noindent
State names of
$\tau_{\OK(G)}(\, \PAClient(\DG) \, )$ have been suppressed in
Figure~\ref{reduce_then_compose}b for readability. Note that the
number of states in $\tau_G( \, \PAClient(\DG) \, )$ is~13, while the
\emph{first-reduce then-compose} approach with $\PAQClient$ and
$\PAClient(\CS)$ generates a process with 9~states only. See
table~\ref{experiments_results} below for more numerical results.

\begin{figure}[hbt!]
  \begin{center}
    \input{Figures/allwrongtogether.pstex_t}
    \caption{(a)~adapted quotient process $\QClient$, (b)~composition
      of new $\PAQClient$ and $\PAClient(\CS)$.}
    \label{wrong_partition}
  \end{center}
\end{figure}

\blankline

\noindent
It is obvious that not every choice of actions at the detailed level
has the property of Lemma~\ref{lemma_reduce_compose}. For example,
selecting the set of actions $G\mkern1mu{}' = \lc \enter, \, \thank \rc$, yields
a split-up into $\lc \Out, \, \Waiting \rc$ and $\lc \Busy, \, \AtDoor
\rc$ and another reduction, depicted in
Figure~\ref{wrong_partition}a. However, this reduction is not a proper
one as the induced composition of the reduced detailed and the global
behaviour in Figure~\ref{wrong_partition} is not branching bisimilar
with the original composition $\tau_{\OK(G\mkern1mu{}')}(\PAClient(\DG))$.

\begin{figure}[hbt!]
  \begin{center}
    \includegraphics[scale=0.50]{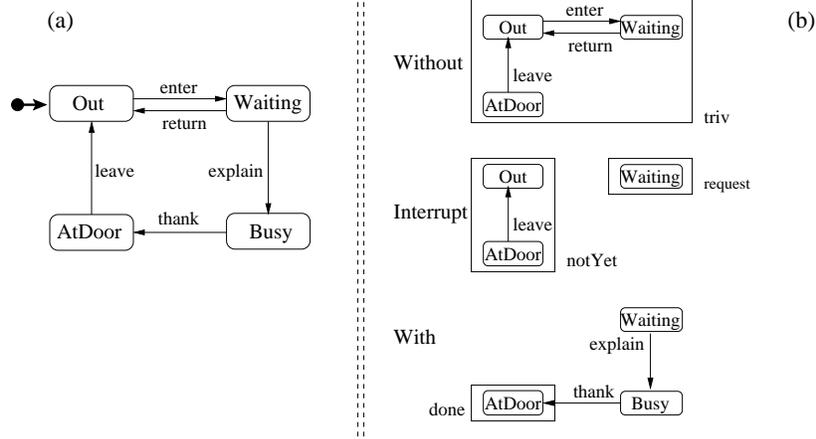}
  \end{center}
  \caption{Modified client: (a)~STD of~$\RClient$, (b)~phase and trap
    constraints.}  \label{Paradigm_client2}
\end{figure}

It is instructive to consider a slightly different client. Now we
assume that the client may decide to draw back the service request and
return back to the initial state $\Out$. The detailed STD and the
global STD shown in Figure~\ref{Paradigm_client2} differ from the
model in Figure~\ref{visualsyntax1} only in the $\return$ transition.
If we apply the same reasoning of Lemma~\ref{lemma_reduce_compose} to
this model of a client, we observe that the $\return$ transition does
not change the situation regarding the reduction of the local
behaviour. Again, the $\enter$ transition is not globally inert, for
the same reasons as in the previous model. Similarly, $\return$ is
also not globally inert. Still, the original quotient from
Lemma~\ref{lemma_reduce_compose} based on the inert actions $\explain$
and $\leave$ yields a proper reduction. See
Figure~\ref{reduce_compose_example2}.

\begin{figure}[bht]
  \begin{center}
    \centerline{\makebox[1.25cm]{}
      \input{Figures/clientDGfromreduced_ex2a.pstex_t}}
    \caption{Branching bisimilar
      processes:~(a)~$\partial_{H}(\RClientQuotient \merge
      \PARClient(\CS))$, (b)~$\tau_{\OK(G)}(\PARClient(\DG))$.}
    \label{reduce_compose_example2}
  \end{center}
\end{figure}

\begin{figure}[bht]
  \begin{center}
    \epsfig{file=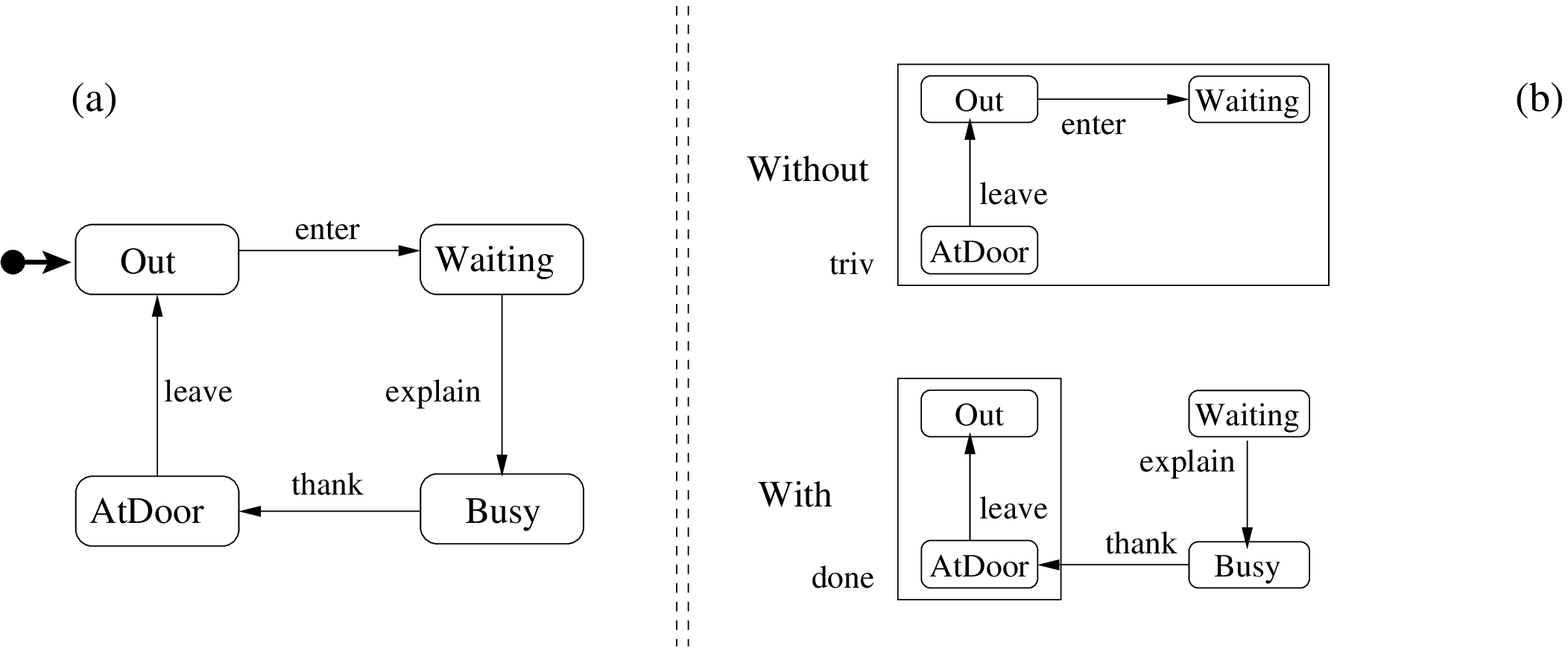,scale=0.40}
    \caption{The Paradigm model of $\SClient$.}
    \label{client3}
  \end{center}
\end{figure}

\blankline

\noindent
The last example we consider as a further variation, named~$\SClient$,
is presented in Figure~\ref{client3}. The only change is now in the
global STD $\SClient(\CS)$. The client is provided service
unconditionally, i.e.\ without interruption, even without needing
it. But, if it doesn't need it the client is handled as if it does not
need service \emph{any longer}. The simplified global behaviour, with
less phases and less traps, imposes less constraints on the detailed
behaviour. Thus, the relation between the detailed and the global
behaviour is rather loose. In Figure~\ref{client3DG} the behaviour of
process $\PASClient(\CS)$ and the parallel composition
$\PASClient(\DG)$ are graphically represented. 
In order to show this formally, we again apply the
\emph{first-reduce then-compose} approach along the lines of
Lemma~\ref{lemma_reduce_compose} by taking the trivial split-up
of~$\StatesD$ along all detailed actions in $\LabelsD$. Thus, we
identify all local actions in $G'' =\LabelsD$ as globally inert. The
resulting quotient STD of $\SClientQuotient$ and its process algebraic
translation are shown in Figure~\ref{client3Pi}bc.  The composition of
the reduced detailed behaviour of $\SClient$ with its global behaviour
has now 3 states as shown in Figure~\ref{client3Pi}d.  A branching
bisimulation between this process and the corresponding process
$\tau_{G''}( \, \PASClient(\DG) \, )$ can be established easily.

\begin{figure}[bht]
  \begin{center}
    \input{Figures/client3All.pstex_t}
    \caption{Processes $\PASClient(\CS)$ and $\PASClient(\DG)$.}
    \label{client3DG}
  \end{center}
\end{figure}

\begin{figure}[thb]
  \begin{center}
    \input{Figures/client3hidden.pstex_t}
    \caption{%
      (a)~$\tau_{G''}(\SClient)$,
      (b)~$\PASClient$,
      (c)~$\PASClientQuotient$,
      (d)~composition of~$\PASClientQuotient$ and~$\PASClient(\CS)$.}
    \label{client3Pi}
  \end{center}
\end{figure}

\blankline

\noindent
In order to investigate the effect of the reduction on a larger scale,
we have analyzed the client-server system using the \texttt{mCRL2}
toolset~\cite{JFGDagstuhl} and compared the implementation of the
system using either the original $\PAClient$ components or their
reduced versions $\PAQClient$.  The translation of ACP-based
specifications of the $n$~clients $\PAClient_i$, the global
$\PAClient_i(\CS)$ and the server $\NDetServer$ into the input
language of the \texttt{mCRL2} toolset, which we use for our model
analysis, is largely straightforward (see also~\cite{scp10}).  Indeed,
the application of the \emph{first-reduce then-compose} principle
yields a significant decrease in the size of the state space in a
number of cases. The results are collected in
Table~\ref{experiments_results}.

\begin{table}[htb]
\begin{center}
\begin{tabular}{|c|rr@{\ \quad}|rr@{\ \quad}|}
\hline
n\raisebox{-5pt}{\rule{0pt}{20pt}} &
\multicolumn{2}{|c|}{with $\PAClient$} &
\multicolumn{2}{|c|}{with $\PAQClient$}
\\
\hline\hline
  & states & \multicolumn{1}{c|}{transitions}
  & states & \multicolumn{1}{c|}{transitions}
  \\
\cline{2-5}
 & & & &
 \\[-0.3cm]

2   & 69     & 142       &  32    & 54
\\
3   & 297   & 819        & 92     & 204
\\
4   & 1161  & 3996      & 240    & 656
\\
5   & 4293  & 17685     & 592   & 1920
\\
6\raisebox{-5pt}{\rule{0pt}{10pt}}
    & 15309 & 73386     & 1408  & 5280
\\
10\raisebox{-5pt}{\rule{0pt}{10pt}}
    & -- & --    & 36863  & 212480
    \\
\hline
\multicolumn{5}{c}{%
  \rule{0pt}{20pt} \small
  (no result for $\PAClient$ with $n{=}10$ within 24 hours)}
\end{tabular}
\caption{Effect of the \emph{first-reduce then-compose} approach.}
\label{experiments_results}
\end{center}
\end{table}

\subsection{Extracting detailed behaviour}
\label{subsec:4-2}

Intuitively it is clear that in the case of the client-server example
the global behaviour does not change or influence the local behaviour.
In fact, if in the total client behaviour $\PAClient(\DG)$ we hide the
actions $\trap(\cdot)$ from the set~$E$ performed by the global
process ($E$ for external), we obtain a process which is branching
bisimilar to the detailed behaviour $\Client$. This is expressed by
the following lemma.

\begin{lem}
  $\Client \bbisim \tau_E( \, \PAClient(\DG) \, )$.
\end{lem}

\begin{proof}
  We start from the process $\PAClient(\DG)$ as shown in
  Figure~\ref{parallel_client_and_clientCS}. After hiding the actions
  in~$E$, i.e.\ renaming them into $\tau$, the process $\tau_E( \,
  \PAClient(\DG) \,)$ is obtained, shown in
  Figure~\ref{hidden_parallel_client_and_clientCS}. A branching
  bisimulation equivalence between this process and $\Client$ process
  can be defined without difficulty. In
  Figure~\ref{bisim_client_and_clientDG} related states are connected
  by differently dotted lines. Note, we have mirrored the $\Client$
  orientation with respect to the North-East South-West diagonal.
\end{proof}

\begin{figure}[bht!]
\begin{center}
\centerline{\input{Figures/clientDGhideglobal.pstex_t}}
\caption{Process $\tau_E(\PAClient(\DG))$.}
\label{hidden_parallel_client_and_clientCS}
\end{center}
\end{figure}

\begin{figure}[htb]
\centerline{\input{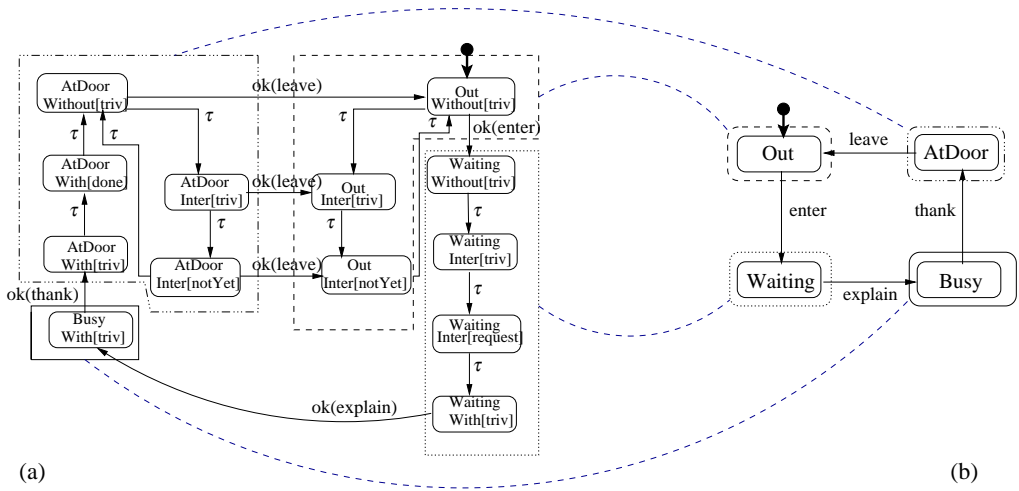}}
\caption{Branching bisimulation between (a)~$\tau_E(\PAClient(\DG))$
  and (b)~$\Client$.}
\label{bisim_client_and_clientDG}
\end{figure}

\noindent
In the general situation, the statement of the lemma provides a check
on the constraints imposed by the global STD on the detailed one. In
case the statement of the lemma holds, the complete behaviour of the
component is preserved in the consistent composition, assuming the
coordinating protocol provides all phase transfers in some order. In
case the statement of the lemma does not hold, part of the original
detailed behaviour has been eliminated because of the participation
with the protocol. This may be deliberate and allows for further
reduction of the detailed STD\@. This may be accidental, requiring the
overall coordination to be revised.


\section{Concluding remarks}
\label{sec:5}

In a Paradigm model several STDs may belong to the same component,
describing the component's dynamics either at various levels of
abstraction (detailed vs. global STDs) or describing different roles
of the component in various collaborations. Collaboration between
components is described in terms of dynamic constraints.  Vertical
consistency is maintained by keeping phases vs.\ detailed transitions
and traps vs.\ transfers aligned. Starting point of our investigation
here is the translation of Paradigm models into the process algebra
ACP and its coupling with the \texttt{mCRL2} toolset for subsequent
automated analysis. In the translated model, every STD from the
Paradigm model is represented by a recursive specification; the total
behaviour of a single component is obtained as a composition of the
recursive specifications of the detailed and the global component's
STDs; the overall system is specified as a parallel composition of all
components.

In this paper we have described a method to reduce the Paradigm
representation of the detailed STDs of the components, yielding
reduction of the overall Paradigm models, but preserving the overall
behaviour. The reduction boils down to inferring globally inert
detailed steps. By abstracting them away a smaller representation of
the detailed component is obtained. This representation contains all
information about the constraints the detailed behaviour imposes on
the global behaviour(s) of the component. The formal validation that
the reduction, indeed, does not change the overall model behaviour is
achieved via the process algebraic representation of the model: we
show for our client-server example that the reduced model is branching
bisimilar to the original one, having the same
properties. Furthermore, by means of a proper abstraction, in this
case applied at the global level, we can observe directly from the
model, by a direct comparison, in which way the global behaviour, and
thus the collaboration, affects the components' detailed behaviour. In
case no influence is to be expected, it is sufficient to show that the
component model is equivalent, up to branching bisimulation, to the
detailed behaviour after all global steps are abstracted away.

As to the contribution of this paper, we have established a further
connection of process algebra and its supporting apparatus to the
domain of coordination. In particular, abstraction and equivalences,
typical for process algebra, become techniques that can be applied to
coordination models, via the established link of the Paradigm language
and ACP, in our case. Thus, coordination can be initially modeled in
the Paradigm language which offers compositional and hierarchical
modeling flexibility. Then, model reduction can be applied, if
appropriate. Finally, via its process representation the model can be
formally analyzed.

As future work we want to address the reduction of general Paradigm
models and property guided reduction, in particular in a situation
with overlapping or orthogonal coordination. More specifically, it is
interesting to study the notion of globally inert detailed steps for a
component that participates in multiple collaborations. We plan to
investigate whether other techniques from process algebraic analysis,
e.g.\ iterated abstraction, and pattern-based simplifications can be
beneficial for the modeling with Paradigm.

\bibliographystyle{eptcs}
\bibliography{paco2011}

\end{document}